\documentclass[a4paper,UKenglish,cleveref, thm-restate]{lipics-v2021}
\nolinenumbers
\title{Identity Testing for Stochastic Languages}

\author{Smayan Agarwal}{Ashoka University}{smayan.agarwal_ug25@ashoka.edu.in}{}{}

\author{Shobhit Singh}{Ashoka University}{shobhit.singh@ashoka.edu.in}{}{}

\author{Aalok Thakkar}{Ashoka University}{aalok.thakkar@ashoka.edu.in}{}{}

\authorrunning{Agarwal, Singh \& Thakkar}

{
} %TODO mandatory: Please choose ACM 2012 classifications from https://dl.acm.org/ccs/ccs_flat.cfm 

 %TODO mandatory; please add comma-separated list of keywords

\usepackage[utf8]{inputenc}
\usepackage[ruled,linesnumbered]{algorithm2e}
\usepackage{amssymb}
\usepackage{stmaryrd}
\usepackage[inline]{enumitem}
\usepackage{tikz}
\usetikzlibrary{automata, positioning, shapes.multipart}
\usepackage{float}
\usetikzlibrary{automata, positioning, arrows.meta}
\newcommand{\semantics}[1]{{\llbracket #1 \rrbracket}}
\newcommand{\calA}{\mathcal{A}}

\newcommand{\N}{\mathbb{N}}

\newcommand{\R}{\mathbb{R}}
\newcommand{\Stoch}{\mathcal{S}}

\newcommand{\pseries}[2]{#1\langle\langle #2 \rangle\rangle}
\hideLIPIcs 
\begin{document}

\maketitle

\begin{abstract}
\label{sec:abstract}
Determining whether an unknown distribution matches a known reference is a cornerstone problem in distributional analysis. While classical results establish a rigorous framework in the case of distributions over finite domains, real-world applications in computational linguistics, bioinformatics, and program analysis demand testing over infinite combinatorial structures, particularly strings. In this paper, we initiate the theoretical study of identity testing for stochastic languages, bridging formal language theory with modern distribution property testing.

We first propose a polynomial-time algorithm to verify if a finite state machine represents a stochastic language, and then prove that rational stochastic languages can approximate an arbitrary probability distribution. Building on these representations, we develop a truncation-based identity testing algorithm that distinguishes between a known and an unknown distributions with sample complexity $\widetilde{\Theta}\left( \frac{\sqrt{n}}{\varepsilon^2} + \frac{n}{\log n} \right)$ where $n$ is the size of the truncated support. Our approach leverages the exponential decay inherent in rational stochastic languages to bound truncation error, then applies classical finite-domain testers to the restricted problem.

This work establishes the first identity testing framework for infinite discrete distributions, opening new directions in probabilistic formal methods and statistical analysis of structured data.
\end{abstract}

\section{Introduction}
\label{sec:introduction}
The problem of identity testing, determining whether an unknown distribution matches a known reference distribution has received significant attention in theoretical computer science. Classically, identity testing considers distributions over finite domains, where the goal is to distinguish between the case where two distributions are identical versus $\varepsilon$-far in total variation distance, using as few samples as possible from the unknown distribution \cite{batu2001testing}. Optimal algorithms for this problem achieve sample complexity of $O(\sqrt{n}/\varepsilon^2)$ for distributions over domains of size $n$ \cite{paninski2008simple,diakonikolas2016collision}, and this framework has been extended to various structured settings including monotone, log-concave, and mixture distributions \cite{canonne2015testing,diakonikolas2015testing}.

However, many applications in computational linguistics, bioinformatics, and program analysis involve distributions over infinite combinatorial structures, particularly strings. For instance, natural language models assign probabilities to sequences of words: including feed-forward neural network models \cite{bengio2003neural}, recurrent neural network language models \cite{mikolov2010recurrent}, and hidden Markov models \cite{rabiner1989tutorial}. DNA analysis involves distributions over nucleotide sequences modeled by probabilistic sequence models \cite{durbin1998biological}, and probabilistic programs generate distributions over execution traces represented as strings \cite{wingate2011lightweight,goodman2008church}. These settings call for a theory of identity testing that can handle the infinite, structured nature of string distributions while maintaining computational tractability \cite{canonne2020survey}.

In this paper, we initiate the study of identity testing for stochastic languages. We focus on the subclasses of geometric and rational stochastic languages which admit finite representations. Our approach unifies techniques from formal language theory and distribution property testing, and extends classical results from finite-domain distributional testing to the infinite and structured domain of formal languages. In particular:

\begin{enumerate}
    \item We propose a polynomial time algorithm to check if a cost register automata over $\R^+$ with linear updates induces a probability distribution.
    \item We prove that convex combinations of geometric distributions as well as rational stochastic languages are dense in the space of all stochastic languages under the $\ell_1$ norm (\cref{thm:dirac-approx}). This allows us to approximate arbitrary distributions with finite representations.
    \item We introduce a succinct and expressive representation for stochastic languages called stochastic regular expressions that can be used for approximate sampling, analysis, and testing of stochastic languages.  
    \item We develop a truncation-based identity testing procedure that distinguishes between close and far distributions under the $\ell_1$ distance with sample complexity $\widetilde{\Theta}\left( \frac{\sqrt{n}}{\varepsilon^2} + \frac{n}{\log n} \right)$ where $n$ is the size of the truncated support(\cref{subsec:algorithm}).
\end{enumerate}

The paper is organized as follows. \cref{sec:relatedworks} reviews related work in distributional testing and weighted automata. \cref {sec:notation} establishes notation and preliminaries. \cref{sec:models} develops the computational framework of cost register automata for stochastic languages. \cref{sec:representation} introduces stochastic regular expressions and proves the universal approximation theorem. \cref{sec:testing} presents our identity testing algorithm and analyzes its correctness and sample complexity. \cref{sec:conclusion} concludes with directions for future work.
\section{Related Work}\label{sec:relatedworks}

The related work in this area can be broadly divided into two directions: distribution property testing and stochastic language. 
\subsection{Stochastic Languages}
\label{subsec:stochastic-lang}

Stochastic languages provide a formal framework for representing probability distributions over strings, with foundational models including weighted automata \cite{DrosteK19} and probabilistic grammars \cite{DenisE062}. These representations bridge formal language theory and statistical learning, enabling applications in grammar induction \cite{Clark2010}, natural language processing \cite{Cohn2009}, and biological sequence modeling \cite{Durbin1998}. 

A fundamental limitation is the undecidability of equivalence checking for stochastic languages, even for simple subclasses \cite{Almagor2022, Almagor2022-2, DrosteK19}. This has motivated two research directions: (1) developing approximation algorithms with probabilistic guarantees \cite{Tulsiani2021}, and (2) identifying decidable subclasses via restrictions like linear recurrence \cite{Weiss2021}.

Spectral learning methods have emerged as a principled alternative to EM-based estimation, offering finite-sample PAC bounds \cite{Balle2014}, handling of missing data via tensor completion \cite{Cohen2017}, and extensions to conditional distributions \cite{Yuan2022}. These advances support applications requiring interpretable models with certified convergence, such as policy learning in robotics and clinical text analysis \cite{Johnson2016}.

\subsection{Distribution Property Testing}

There has been extensive progress in the area of distribution property testing over the past three decades. The formal study of distribution testing began with Goldreich, Goldwasser, and Ron \cite{Goldreich1998PropertyTesting}. Batu et al. \cite{batu2001thesis,batu2001testing,Batu2001} initiated the systematic study of closeness and identity testing, introducing foundational techniques such as bucketing. Identity testing has a tight sample complexity of $\Theta(\sqrt{k}/\varepsilon^2)$, established through the work of Chan et al. \cite{ChanDiakonikolasValiant2013}, Acharya et al. \cite{acharya2015optimal,acharya2011competitive,acharya2015histograms}. Variants include the $\chi^2$-based tester, an $\ell_2$-based reduction \cite{diakonikolas2016new}, and a reduction from identity to uniformity \cite{chakraborty2018conditional}. The idea of instance-optimal identity testing was introduced by Valiant and Valiant \cite{valiant2017automatic}, showing that sample complexity can depend on the specific reference distribution, not just the domain size. Tolerant testing has also been studied, with recent lower bounds and separations established by Canonne \cite{Canonne2022Price,chakraborty2022gap}

Contemporary works frequently leverage conditional sampling \cite{chakraborty2018conditional,ChakrabortyEtAl2012}, which has led to dramatic improvements in sample complexity. Surveys by Diakonikolas \cite{DiakonikolasKane2016} and Canonne \cite{canonne2020survey,canonne2022topics} provide comprehensive overviews of these developments. Our work initiates the study of identity testing over $\Sigma^+$ by combining automata-based representations, truncation arguments, and conditional sampling to design testers for infinite discrete distributions.

Despite this progress, the work is limited to distribution with finite support. In contrast, property testing over infinite discrete domains like strings $\Sigma^+$ has received very little attention.

\section{Preliminaries and Notation}
\label{sec:notation}

In this section, we recall some definitions and results about probability distributions and strings. Recall that if $\Sigma$ is a finite set, then $\Sigma^*$ denotes the set of all words with $\Sigma$ as the alphabet and $\varepsilon$ as the empty word. That is, $\Sigma^*$ is a freely and finitely generated monoid with $\varepsilon$ as the unit and concatenation as the operator. 

\begin{definition}[Semiring]
A {\em semiring} is a set $K$ with two binary operations $+$ and $\cdot$ and two constant elements $0$ and $1$ such that: 
\begin{enumerate}
    \item $(K, +, 0)$ is a commutative monoid,
    \item $(K,\cdot, 1)$ is a monoid, 
    \item The distribution laws $a \cdot (b +c) = a \cdot b  + a \cdot c$ and $(a + b) \cdot c = a \cdot c +b \cdot c$ hold,
    \item $0 \cdot a = a \cdot 0 = 0$ for every $a \in K$.
\end{enumerate}
\end{definition}

\begin{definition}[Formal Power Series]
Let $\Sigma$ be a finite alphabet and $K$ a semiring. A formal power series is a mapping $r: \Sigma^* \to K$. The values $r(w)$ where $w \in \Sigma^*$ are referred to as the coefficients of the series, and $r$ is written as a formal sum $r = \sum_{w \in \Sigma^*} r(w)w$. The set of all formal power
series is denoted by $K\langle\langle \Sigma \rangle\rangle$. The support of $r$ is $\{w : r(w) \neq 0 \}$.
\end{definition}

The sum of two series $r, r' \in \pseries{K}{\Sigma}$ is defined by $r + r' = \sum_{w \in \Sigma^*} (r(w) + r'(w)) w$. The scalar multiplication of a series $r$ with a $\lambda \in K$ is defined as $\lambda r = \sum_{w \in \Sigma^*} (\lambda r(w))w$.  The Cauchy product of two series $r$ and $r'$ is defined by $r \cdot r'  =\sum_{w \in \Sigma^*, w_1w_2=w} (r(w_1) \cdot r'(w_2)) w$. We now look at a certain class of formal power series that corresponds to probability distributions over strings. 

\begin{definition}[Stochastic Language]\label{def:stochastic-language}
A stochastic language is a formal series $r \in \pseries{\R^+}{\Sigma}$ such that
$\sum_{w \in \Sigma^*} r(w) = 1$. For any stochastic language $r$ and any language $L \subseteq \Sigma^*$, the $\sum r(L) = \sum_{w \in L} r(w)$ is well defined.  In this paper, we work with stochastic languages over $\Sigma^+$ (that is, we do not consider non-empty strings), and denote the set of all stochastic languages over $\Sigma^+$ by $\Stoch(\Sigma)$.

\end{definition}

We find it useful to define $\mathcal{M}(\Sigma) \subseteq \pseries{\R^+}{\Sigma}$ such that: 
\begin{equation*}
    \mathcal{M}(\Sigma) := \left\{ \mu : \Sigma^+ \to \R^+ \;\middle|\; \sum_{w \in \Sigma^+} \mu(w) < \infty \right\}
\end{equation*}   

That is, $\mathcal{M}(\Sigma)$ contains all formal power series such that the sum of the coefficients is bounded. We can scale coefficients of $\mu$ so that they sum up to one. We call this scalar transformation $\overline{\mu} = \frac{\mu}{\mu(\Sigma^+)}$ a normalisation of $\mu$. Clearly, $\forall \mu \in \mathcal{M}(\Sigma) \backslash \{0\}$, $\overline{\mu} \in \Stoch(\Sigma)$. 
\begin{example}\label{ex:zeta-example}
Let $\Sigma = \{a\}$ be a singleton alphabet. Define the formal power series $\mu \in \pseries{\R^+}{\Sigma}$ as $\mu:w \mapsto \frac{1}{|w|^2}$. This series satisfies:
\begin{equation*}
    \sum \mu(\Sigma^+) = \sum_{w \in \Sigma^+} \mu(w) = \sum_{n=1}^{\infty} \mu(a^n) = \sum_{n=1}^{\infty} \frac{1}{n^2}  = \frac{\pi^2}{6}.
\end{equation*}
The series $r$ belongs to $\mathcal{M}(\Sigma)$ since its total weight converges to $\pi^2/6$. And therefore we have $\overline{r} : w\mapsto  \frac{6}{\pi^2}\frac{1}{|w|^2}\in \Stoch(\Sigma)$
\end{example}

\begin{example}\label{ex:poisson-power-series}
Let $\Sigma = \{a,b,c\}$ and $\lambda \in (0, 1)$. Then consider the power series $\mu \in \pseries{\mathcal M}{\Sigma}$ defined as $
\mu:w \mapsto \frac{\lambda^{|w|}}{|w|! \cdot 3^{|w|}}$.
\begin{equation*}
    \sum \mu(\Sigma^+) = \sum_{w \in \Sigma^+} \mu(w) = \sum_{w \in \Sigma^+} \frac{\lambda^{|w|}}{|w|! \cdot 3^{|w|}} = \sum_{n = 1}^\infty \Biggr(\sum_{w \in \Sigma^n} \frac{\lambda^{n}}{n! \cdot 3^{n}}\Biggl) = \sum_{n =1}^\infty \frac{\lambda^{n}}{n!} = e^{\lambda} - 1
\end{equation*}

And therefore, we have $\overline{\mu} : w \mapsto \frac{1}{e^\lambda - 1} \frac{\lambda^{|w|}}{|w|! \cdot 3^{|w|}} \in \Stoch(\Sigma)$. 
\end{example}

A central subclass of $\Stoch(\Sigma)$ consists of distributions with finite support. In particular, there is the Dirac distribution, which concentrates all probability on a single string. 

\begin{definition}[Dirac Distribution]
For a fixed string $w \in \Sigma^+$, the \emph{Dirac distribution} $\delta_w$ is a probability distribution over $\Sigma^+$ defined for all $u \in \Sigma^+$ as:
\[
\delta_w(u) = 
\begin{cases} 
1 & \text{if } u = w, \\
0 & \text{otherwise},
\end{cases}
\]
\end{definition}

Dirac distributions serve as the atomic unit of all finite distributions over $\Sigma^*$. Every finite-support distribution can be decomposed into a convex combination of Dirac distributions. 
\section{Finitary State Models for Stochastic Languages}
\label{sec:models}
In general, several rational power series satisfy \cref{def:stochastic-language}. We are interested in those that have {\em finitary} representations amenable to analysis, with sufficiently nice computational properties. To that end, we consider the cost register automata model \cite{AlurDDRY13}. 

\begin{definition}[Cost Register Automata (CRA)] A cost register automaton $\calA$ is a tuple $\calA=(\Sigma, Q, X, q_0, x_0, \delta, \rho, \mu)$, where
\begin{enumerate}
    \item $\Sigma$ is a finite alphabet,
    \item $Q$ is a finite set of states,
    \item $X$ is a finite set of registers,
    \item $q_0$ is an initial state,
    \item $x_0: X \to \mathbb{R}$ is an initial valuation of the registers,
    \item $\delta\colon Q \times \Sigma \to Q$ is the state transition function,
    \item $ \rho \colon Q \times \Sigma \times \mathbb{R}^X \to \mathbb{R}^X $ is the register update function,
    \item $\mu: Q \times \mathbb{R}^X \to \mathbb{R}$ is the finalization function.
\end{enumerate}

The {\em semantics} of a cost register automaton $\calA$ is a function $\semantics{\calA}\colon \Sigma^+ \to \mathbb{R}$ defined as follows. Given a word $w=w_1 \ldots w_n \in \Sigma^+$, the associated run of the automaton is a sequence of states and register valuations $(q_0,x_0), \ldots, (q_n, x_n)$ such that:
\begin{enumerate}
    \item $q_0$ and $x_0$ are the initial state and the initial register valuation,
    \item $q_i = \delta(q_{i-1}, w_i)$, and
    \item $x_i = \rho(q_{i-1}, w_i, x_{i-1})$ for $1 \leq i \leq n$.
\end{enumerate}
The semantics is then given by $\semantics{\calA}(w)=\mu(q_n,x_n)$.
\end{definition}

\begin{figure}
    \centering
    \begin{tikzpicture}[scale=0.2]
        \tikzstyle{every node}+=[inner sep=0pt]
        \draw [black] (20.2,-32.8) circle (3);
        \draw [black] (20.2,-32.8) circle (2.4);
        \draw [black] (10.4,-32.8) -- (17.2,-32.8);
        \fill [black] (17.2,-32.8) -- (16.4,-32.3) -- (16.4,-33.3);
        \draw [black] (22.88,-31.477) arc (144:-144:2.25);
        \draw (27.45,-31.8) node [right] {$\sigma: x \leftarrow x + 1$};
        \draw (27.45,-33.8) node [right] {$\sigma: y \leftarrow \frac{\lambda y}{|\sigma|x}$};
        % \draw (27.45,-33.8) node [right] {$\sigma: x_g \leftarrow x_g y_g + g(\sigma)$};
        % \draw (27.45,-35.8) node [right] {$\sigma: y_g \leftarrow 10^{|g(\sigma)|}$};
        \fill [black] (22.88,-34.12) -- (23.23,-35) -- (23.82,-34.19);
        \draw (15.45,-38.8) node {$\displaystyle \rho(x)=0$};
        \draw (27.45,-38.8) node {$\displaystyle \rho(y)=e^{-\lambda}$};
        \draw (37.45,-38.8) node {$\displaystyle \mu=y$};

    \end{tikzpicture}
    \caption{A Cost Register Automata for the stochastic language in \cref{ex:poisson-power-series}.}
    \label{fig:cra-poisson}
\end{figure}
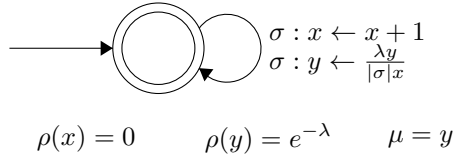

A cost register automaton induces a probability distribution if its semantics satisfies \cref{def:stochastic-language}. In general, checking whether a cost register automaton induces a probability distribution is undecidable, as the following lemma shows.

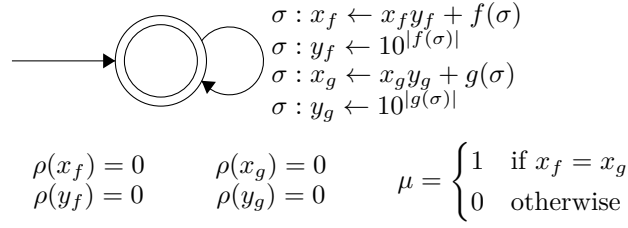
\begin{figure}
    \centering
    \begin{tikzpicture}[scale=0.2]
        \tikzstyle{every node}+=[inner sep=0pt]
        \draw [black] (20.2,-32.8) circle (3);
        \draw [black] (20.2,-32.8) circle (2.4);
        \draw [black] (10.4,-32.8) -- (17.2,-32.8);
        \fill [black] (17.2,-32.8) -- (16.4,-32.3) -- (16.4,-33.3);
        \draw [black] (22.88,-31.477) arc (144:-144:2.25);
        \draw (27.45,-29.8) node [right] {$\sigma: x_f \leftarrow x_f y_f + f(\sigma)$};
        \draw (27.45,-31.8) node [right] {$\sigma: y_f \leftarrow 10^{|f(\sigma)|}$};
        \draw (27.45,-33.8) node [right] {$\sigma: x_g \leftarrow x_g y_g + g(\sigma)$};
        \draw (27.45,-35.8) node [right] {$\sigma: y_g \leftarrow 10^{|g(\sigma)|}$};
        \fill [black] (22.88,-34.12) -- (23.23,-35) -- (23.82,-34.19);
        \draw (15.45,-39.8) node {$\displaystyle \rho(x_f)=0$};
        \draw (15.45,-41.8) node {$\displaystyle \rho(y_f)=0$}; 
        \draw (27.45,-39.8) node {$\displaystyle \rho(x_g)=0$};
        \draw (27.45,-41.8) node {$\displaystyle \rho(y_g)=0$};
        \draw (43.45,-40.8) node {$\displaystyle \mu=\begin{cases}
            1 & \text{if } x_f = x_g\\
            0 & \text{otherwise}
        \end{cases}$};
    \end{tikzpicture}
    \caption{Cost register automaton used in the proof of \cref{lem:cra-undecidability}}
    \label{fig:cra-undecidability}
\end{figure}
\begin{lemma}\label{lem:cra-undecidability}
    Given a CRA $\mathcal{A}$, it is undecidable if $\sum_{w\in \Sigma^+}\semantics{\mathcal{A}}(w) = 1$.
\end{lemma}

\begin{proof}
The proof proceeds through a reduction from the Post Correspondence Problem (PCP). Let the given instance of PCP be with $k$ tiles and alphabet $ \Sigma = \{1, \ldots 9\}$ with maps $f, g : \{1, \ldots, k \} \rightarrow \Sigma^+$. One can interpret $f$ and $g$ as mapping sequence of tiles to natural numbers expressed in their decimal form. Then consider the CRA in \cref{fig:cra-undecidability}, with one state and four registers $x_f$, $y_f$, $x_g$ and $y_g$ all initialized to $0$. At each step, $x_f$ keeps track of $f(w)$ and $y_f$ keeps track of ten to the power of the size of the last tile used. $x_f$ is updated to $x_f y_f + f(\sigma)$ on reading the letter $\sigma$. $x_g$ and $y_g$ function analogously. $\mu$ is defined as $\textsf{if } x_f = x_g \textsf{ then } 1 \textsf { else } 0$. 

Observe that if the PCP instance has a solution $w$, then the CRA maps $w$ to 1. Also, if $w$ is a solution, then so is $w^i$ for all $i$, and hence infinitely many strings are mapped to $1$. In that case, the CRA does not induce a probability distribution. Conversely, if there is no solution to the PCP instance, then only the empty string is mapped to $1$ and everything else is mapped to $0$, which is a probability distribution.
\end{proof}

Hence we consider restrictions on CRAs. It is suitable, although not sufficient, to restrict $\rho$ and $\mu$ to linear functions of register valuations. 
\subsection{Rational Stochastic Languages}

It is well known that CRAs with linear updates are equivalent to weighted automata \cite{AlurDDRY13}\cite{Benalioua2024}\cite{Czerwinski2022}. Recall the definition of a weighted automata: 

\begin{definition}[Weighted Automata]\label{def:weighted}
    Let $\Sigma$ be a finite alphabet and $(K, +, \cdot, 0, 1)$ be a semiring. A weighted automaton $\calA$ is a tuple $\calA=(\Sigma, Q, \lambda, \rho, \mu)$, where:
    \begin{enumerate}
        \item $\Sigma$ is a finite alphabet
        \item $Q$ is a finite set of states,
        \item $\lambda \colon Q \to K$ is the initial weight function,
        \item $\mu \colon Q \to K$ is the final weight function,
        \item $\rho\colon Q \times \Sigma \times Q \to K$ is the transition weight function.
    \end{enumerate}
    
    A path $\pi$ in $\calA$ is a sequence $q_0 \xrightarrow{\sigma_1} q_1 \xrightarrow{\sigma_2} \cdots \xrightarrow{\sigma_n} q_n$. Its weight is:
    \[ w(\pi) = \lambda(q_0) \cdot \left(\prod_{i=1}^n \rho(q_{i-1}, \sigma_i, q_i)\right) \cdot \mu(q_n) \]
    and its \emph{label} is $\ell(\pi) = \sigma_1 \cdots \sigma_n \in \Sigma^+$. The \emph{semantics} of $\calA$ is the function $\semantics{\calA}\colon \Sigma^+ \to K$:
    \begin{equation*}
        \semantics{\calA}(u) = \sum_{\substack{\pi \text{ path}\\ \ell(\pi)=u}} w(\pi) 
    \end{equation*}
    where the sum uses the semiring's addition operation.
\end{definition}

The stochastic languages represented by weighted automata (and CRAs with linear updates) are called rational stochastic languages. 

\begin{definition}[Rational Stochastic Language]\label{def:stochastic-rational}
The class of rational stochastic languages over alphabet $\Sigma$ and field $\R$, denoted $\Stoch_{\R}^{rat}(\Sigma)$, consists of all functions $f \colon \Sigma^+ \to \R$ such that:

\begin{enumerate}
    \item $f = \semantics{\mathcal{A}}$ for some weighted automaton $\mathcal{A}$ over the semiring $(\R, +, \cdot, 0, 1)$
    \item $f$ is a stochastic language 
\end{enumerate}

The restricted class $\Stoch_{\R^+}^{rat}(\Sigma)$ considers only automata with non-negative weights. When the semiring is clear from context, we write simply $\Stoch^{rat}(\Sigma)$.
\end{definition}

\begin{lemma}[\cite{abk20}]\label{thm:lin-cra-undecidability}
    Given a CRA $\calA$ with linear updates, it is undecidable whether $\semantics{A}(w)\geq 0$ for every word $w\in \Sigma^+$.
\end{lemma}

There is an easy way around \cref{thm:lin-cra-undecidability}: we can impose a further restriction that the initial values of the registers, the update matrices, and the final vectors should all be non-negative. Equivalently, for weighted automata we can require all the weights to be non-negative, thus ensuring that condition (1) in \cref{def:stochastic-language} is trivially satisfied. Checking condition (2) is then reduced to solving a certain system of linear equations.

\begin{remark}
    By restricting to the non-negative models in $\Stoch_{\R^+}^{rat}(\Sigma)$, we give up a bit of expressivity. Indeed, $\Stoch_{\R^+}^{rat}(\Sigma) \subsetneq \Stoch_{\R}^{rat}(\Sigma) \subsetneq \Stoch(\Sigma)$ \cite{DenisE06}.
\end{remark}

In \cref{subsec:total-weight}, we show that $\sum_{w \in \Sigma^+}\semantics{\calA}$ can be computed if $\calA$ is a CRA with linear updates. A related, seemingly more expressive model would be one where we allow the update and finalisation functions to be affine. However, the expressivity of the affine model is the same as the expressivity of the linear model (refer to \cref{thrm:app:affine}). Unfortunately, it turns out, even with only linear updates, the question of checking if a CRA induces a probability distribution remains undecidable.
\subsection{Computing the Total Weight via Linear Systems}
\label{subsec:total-weight}

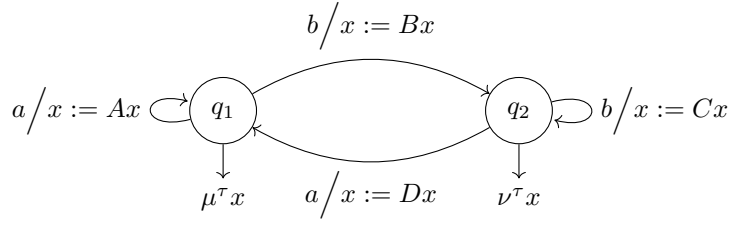
\begin{figure}
    \centering
    \begin{tikzpicture}[->,auto,node distance=3cm,accepting/.style=accepting by arrow]
        \node[state,accepting,accepting text=$\mu^\tau x$,accepting where=below] (q_1) {$q_1$};
        \node[state,accepting,accepting text=$\nu^\tau x$,accepting where=below] (q_2) [right=of q_1] {$q_2$};
        \path (q_1) edge [loop left]
                    node {$a \Big/ x := Ax$} ()
                    edge [bend left]
                    node {$b \Big/ x := Bx$} (q_2)
              (q_2) edge [loop right]
                    node {$b \Big/ x := Cx$} ()
                    edge [bend left]
                    node {$a \Big/ x := Dx$} (q_1);
    \end{tikzpicture}
    \caption{A cost register automaton with two states over an alphabet $\Sigma=\{a, b\}$}
    \label{fig:abcd}
\end{figure}
\begin{figure}[t]
    \centering
    \begin{minipage}[c]{0.4\textwidth}
        \centering
        \begin{tikzpicture}[->,auto,accepting/.style=accepting by arrow,
            scale=0.9, transform shape]
            \node[state,accepting,accepting where=right,%
                  accepting text=$\mu^\tau x$,%
                  initial,%
                  initial text=$x_{\mathsf{init}}$] (q) {$q$};
            \path (q) edge [loop above]
                        node [align=center] {$a \big/ x :=  Ax$ \\ $b \big/ x := Bx$} (q);
        \end{tikzpicture}
    \end{minipage}%
    \begin{minipage}[c]{0.3\textwidth}
        \centering
        \[
            \begin{array}{c}
                A = \dfrac{1}{k} \begin{bmatrix} 1 & 1 \\ 0 & 1 \end{bmatrix} \\
                B = \dfrac{1}{k} \begin{bmatrix} 1 & 0 \\ 0 & 1 \end{bmatrix}  % Removed extra line break
            \end{array}
        \]
    \end{minipage}%  % Added comment to prevent space
    \begin{minipage}[c]{0.3\textwidth}
        \centering
        \[
            \begin{array}{c}
                x_{\mathsf{init}} = \begin{bmatrix} 0 \\ 1 \end{bmatrix} \\
                \mu = \begin{bmatrix} 1 \\ 0 \end{bmatrix}
            \end{array}
        \]
    \end{minipage}
    \caption{Cost register automaton from \cref{ex:count-as}.}
    \label{fig:count-as}
\end{figure}
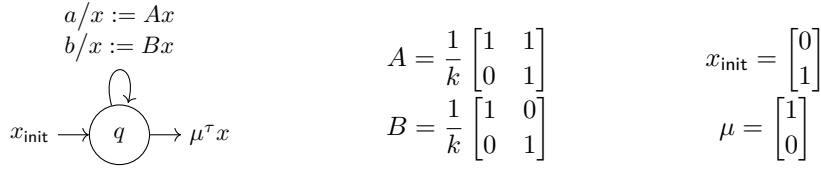
We now show that it is possible to determine if a CRA over $\R^+$ with linear updates induces a stochastic language. We first demonstrate this with an example of the CRA in \cref{fig:abcd}. 

Let us define the following quantities for $i=1,2$ and an arbitrary vector $x$:
\[ S_i(x) := \sum_{w \in \Sigma^+} \semantics{\calA_{q_i, x}}(w) \]
Clearly, $S_i(x)$ is linear in $x$. Moreover, by expanding the semantics of the automaton,
we get
\begin{align*}
    S_1(x) & = \mu^\tau x + S_1(Ax) + S_2(Bx) \\
    S_2(x) & = \nu^\tau x + S_1(Dx) + S_2(Cx)
\end{align*}

From \cref{fig:abcd}, observe that $x$, $\mu$, and $\nu$ are $2\times 1$ vectors and $A$, $B$, $C$, and $D$ are $2\times 2$ matrices. Let $e_1$ and $e_2$ be the standard basis vectors of $\mathbb{R}^2$. By specialising the above system for $x=e_1$ and $x=e_2$ and treating $s_{ij}:=S_i(e_j)$ for $i,j\in\{1,2\}$ as unknowns, we get the following system of
equations
\begin{equation} \label{eq:system}
\begin{bmatrix}
    A^\tau-I & B^\tau \\
    D^\tau & C^\tau-I
\end{bmatrix}
\begin{bmatrix}
    s_{11} \\ s_{12} \\ s_{21} \\ s_{22}
\end{bmatrix}
+
\begin{bmatrix}
    \mu \\ \nu
\end{bmatrix}
=0 \, .
\end{equation}
If the sums $S_i(x)$ converge, then \cref{eq:system} has a solution and $S_i(x)=s_{i1}x_1+s_{i2}x_2$. However, \cref{eq:system} can have a solution even
when the sums do not converge. 

Next, we analyze the example in \cref{fig:abcd} for some concrete matrices $A, B, C, D$.

\begin{example}\label{ex:count-as}
    Let $A=D=\frac{1}{k}\left[\begin{smallmatrix} 1 & 1 \\ 0 & 1 \end{smallmatrix}\right]$, $B=C=\frac{1}{k}\left[\begin{smallmatrix} 1 & 0 \\ 0 & 1 \end{smallmatrix}\right]$ for $k>0$, $\mu = \nu = \left[ \begin{smallmatrix} 1 \\ 0 \end{smallmatrix} \right]$, and let the initial register valuation be $\left[ \begin{smallmatrix} 0 \\ 1 \end{smallmatrix} \right]$. It is not difficult to see that in this case the two states $q_1$ and $q_2$ are redundant, and we can simplify the automaton to the one in \cref{fig:count-as}.
    
    First note that $\semantics{\calA}(a^m b^{n-m})=m/k^n$ for $0\leq m \leq n$. Since the matrices $A$ and $B$ commute, we also have $\semantics{\calA}(w)=m/k^n$ for every word $w$ of length $n$ with exactly $m$ symbols $a$. Therefore,
    \[ \sum_{w \in \Sigma^+} \semantics{\calA}(w)
        = \sum_{n \geq 0} \sum_{0 \leq m \leq n} {n \choose m} \frac{m}{k^n}
        = 2 \sum_{n \geq 0} \frac{n 2^n}{k^n} \, .\]
    This sum converges to $k/(k-2)^2$ if and only if $k > 2$. Now let us denote by $s_1$ and $s_2$ the sums of the automaton with unit vectors $e_1$ and $e_2=x_{\mathsf{init}}$ as initial register valuations. Then we have the following system of equations:
    \begin{align*}
        s_1 & = 1 + \frac{1}{k} s_1 + \frac{1}{k} s_1 \\
        s_2 & = \frac{1}{k} s_1 + \frac{1}{k} s_2 + \frac{1}{k} s_2
    \end{align*}
    The solution exists whenever $k\neq 2$, and it is given by $s_1=k/(k-2)$ and $s_2 = k/(k-2)^2$. Note that $s_2$ is always positive. However, if $0<k<2$, then $s_1$ is negative, coinciding with the divergence of the sum $\sum_{w \in \Sigma^+} \semantics{\calA}(w)$, which is supposed to be equal to $s_2$. Note that by setting $k=4$ we get $s_2=1$ (and $s_1=2$), which means that the automaton induces a probability distribution.
\end{example}

Let us now turn to the general case. Let $\mathcal{A} = (Q, q_0,  \{x_1, \dots, x_d\}, \delta, \mu)$ be a CRA with linear updates. 
Let $x_{\mathsf{init}} \in \mathbb{R}^d$ denote the initial register valuation. For each state $q \in Q$, define the function $S_q(x) := \sum_{w \in \Sigma^+} \semantics{\mathcal{A}_{q,x}}(w)$,
where $\mathcal{A}_{q,x}$ denotes the CRA $\mathcal{A}$ with initial state $q$ and initial register valuation $x$. As $\mathcal{A}$ is linear, each $S_q(x)$ is itself a linear function of $x$. That is, for each $q$, there exists a row vector $s_q^\top \in \mathbb{R}^{1 \times d}$ such that $S_q(x) = s_q^\top x$. For each $q \in Q$, we have:

\begin{equation}
\label{eq:total-weight}
S_q(x) = \mu_q^\top x + \sum_{\sigma \in \Sigma} \sum_{\delta(q, \sigma) = (q', A_{q,\sigma})} S_{q'}(A_{q,\sigma} x)    
\end{equation}

By substituting $S_{q'}(A_{q,\sigma} x) = s_{q'}^\top A_{q,\sigma} x$, using linearity, and equating coefficients in the above equation, we get:

\begin{theorem}\label{thm:linear-system-convergence}
Let $\mathcal{A}$ be a CRA with linear updates and finalization, and let $x_{\mathsf{init}} \in \mathbb{R}^d$ be the initial register valuation. 
Let $\{s_q \in \mathbb{R}^d \mid q \in Q\}$ be the unique solution to the system
\[
s_q = \mu_q + \sum_{\sigma \in \Sigma} A_{q,\sigma}^\top s_{q'} \quad \text{where } \delta(q, \sigma) = (q', A_{q,\sigma}).
\]
Then the total weight assigned by $\mathcal{A}$ is given by
\[
\sum_{w \in \Sigma^+} \semantics{\mathcal{A}}(w) = s_{q_0}^\top x_{\mathsf{init}}.
\]
\end{theorem}

If the entries of each $s_q$ are finite and non-negative, then the sum converges and the CRA defines a proper distribution (possibly after normalization). If $s_{q_0}^\top x_{\mathsf{init}} = 1$, then $\mathcal{A}$ defines a stochastic language.

Concretely, we can claim:

\begin{theorem}
Given a CRA $\calA$ with linear updates and a regular language $R \subseteq \Sigma^+$, the quantity $\sum_{w \in R} \semantics{\calA}(w)$ can be computed in $\mathcal{O}\left(|\mathcal{A}|^3\cdot|R|^3\right)$, where $|\mathcal{A}|$ is the size of the CRA and $|R|$ is the size of a deterministic finite state automata (DFA) accepting  $R$.
\end{theorem}

\begin{proof}[Sketch]
Let $\mathcal{B} = (Q_R, \Sigma, \delta_R, q_0^R, F_R)$ be a DFA accepting $R$. We construct the synchronized product of the CRA $\mathcal{A}$ and the DFA $\mathcal{B}$, denoted $\mathcal{A} \times \mathcal{B}$. This product defines the same register semantics as $\mathcal{A}$, but restricted to strings in $R$. That is,
\[
\semantics{\mathcal{A} \times \mathcal{B}}(w) =
\begin{cases}
    \semantics{\mathcal{A}}(w) & \text{if } w \in R, \\
    0 & \text{otherwise}.
\end{cases}
\]

It follows that $\sum_{w \in R} \semantics{\mathcal{A}}(w) =
\sum_{w \in \Sigma^+} \semantics{\mathcal{A} \times \mathcal{B}}(w)$. As shown in Theorem~\ref{thm:linear-system-convergence}, for a CRA with linear updates, the total weight over $\Sigma^+$ can be computed by solving a linear system whose size is proportional to the product of the number of states and the number of registers. In this case, the size of the product automaton is $\mathcal{O}(|\mathcal{A}||R|)$. Solving a system of $n$ linear equations can be done in time $\mathcal{O}(n^3)$ using standard Gaussian elimination. Thus, the overall runtime is $\mathcal{O}(|\mathcal{A}|^3 \cdot |R|^3)$ after pre-processing and symbolic simplification.
\end{proof}
\section{Representations and Approximations for Stochastic Languages}
\label{sec:representation}
While cost register automata (CRAs) can represent a broad class of distributions over strings, including all rational stochastic languages, reasoning about them directly often involves explicit manipulation of matrix-valued register updates. This low-level representation can obscure the underlying structure of the distribution and hinder interpretability.

A key observation is that for large enough strings, the probability assigned to a by a rational stochastic distribution decays exponentially. This arises from the semantics of CRAs: the repeated application of linear updates along transitions leads to a geometric decay in weight. To make this structure more explicit, we first introduce a new algebraic fragment of stochastic languages, which we term geometric distributions. We show that convex combinations of geometric distributions are dense in the space of stochastic languages, making them a natural and expressive basis for approximation and representation.

Building on this insight, we further develop a syntax of stochastic regular expressions (SRE), which is an analogue of the classical regular expressions with probabilistic constructs such as convex combination, Cauchy product, and a discounted Kleene star operator that ensures convergence. We prove that these expressions are equally expressive as rational stochastic languages. This provides a compositional and interpretable representation framework that aligns closely with the semantics of CRAs and weighted automata, while enabling algebraic reasoning about probabilistic string distributions.

\subsection{Geometric Distributions}

Motivated by geometric distributions over $\N$ (which models the waiting time for the first success in Bernoulli trials), we define \emph{ geometric distributions} over $\Sigma^+$, a natural analogue for formal languages. These distributions form a parametric subclass of $\Stoch_{\R^+}^{rat}(\Sigma)$ with particularly simple finite state representations.

\begin{definition}[Geometric Distribution]\label{def:kleene-dist}
Given $w \in \Sigma^+$ and $\alpha \in (0, 1)$, a geometric distribution over $\Sigma^+$ is the function $P_{w}^{\alpha} : \Sigma^+ \to [0,1]$ defined by:
\[
P_{w}^{\alpha}(w') = \begin{cases}
\alpha(1 - \alpha)^{k-1}, & \text{if } \exists k \in \mathbb{N}^+ \text{ such that } w' = w^k, \\
0, & \text{otherwise}.
\end{cases}
\]

We denote the class of all such distributions by $\mathcal{G}(\Sigma) \subset \Stoch_{\R^+}^{rat}(\Sigma)$.
\end{definition}

These distributions succinctly capture the probability of observing $k$ consecutive occurrences of a pattern $w$, which decays exponentially with $k$. Each $P_w^\alpha$ admits a compact representation as a CRA tracking only the count $k$. The normalisation condition $\sum_{k=1}^\infty P_w^\alpha(w^k) = 1$ follows by geometric series convergence. Convex combinations of those geometric distributions provide a succinct and expressive framework for approximating arbitrary stochastic languages. First, we prove that geometric distributions can approximate the simplest distribution:

\begin{lemma}[Geometric Distributions Can Approximate Dirac Distributions]
\label{thm:dirac-approx}
For all Dirac distributions $\delta_w$, given an $\varepsilon > 0$, there exists a geometric distribution $p(w, \varepsilon) \in \mathcal{G}$ such that $\|\delta_w - P_w^\alpha \|_1 < \varepsilon$. 
\end{lemma}

\begin{proof}

Consider $p = P_w^\alpha$ for any $\alpha \in \left(\frac{\varepsilon}{1+\varepsilon}, 1 \right)$. Then $\left|\delta_w(w) - P_w^\alpha(w)\right| = (1-\alpha)$. For any $u = w^k$ with $k > 1$, $\left|\delta_w(w^k) - P_w^\alpha(w^k)\right| = \alpha(1 - \alpha)^k$. For any other $u$, $\left|\delta_w(u) - P_w^\alpha(u)\right| = 0$. Therefore we have:

\begin{equation*}
    \|\delta_w - P_w^\alpha \|_1 = \alpha + \sum_{k = 2}^\infty \alpha(1-\alpha)^k = \frac{\alpha}{(1 - \alpha)} < \varepsilon 
\end{equation*}

\end{proof}

Recall that every stochastic language over $\Sigma^+$ with a finite support is a convex combinations of Dirac distributions. Therefore, convex combinations of geometric distributions can approximate all stochastic languages with finite support with respect to the $\ell_1$ norm. Using this property, we prove the following: 

\begin{theorem}[Universal Approximation]
\label{thm:universal-sup-approx} Given a stochastic language $r \in \Stoch(\Sigma)$, and $\varepsilon > 0$, there exists $p_1,  \ldots, p_k \in \mathcal{G}$ and $\lambda_1, \ldots, \lambda_k \in [0, 1]$ such that 
$\left\|r - \sum_{i = 1}^k \lambda_ip_i \right\|_1 < \varepsilon$.
\end{theorem}

\begin{proof}
Since $r \in \Stoch(\Sigma)$, there exists a finite set $S \subset \Sigma^+$ such that $\sum_{w \in S} r(w) > 1 - \frac{\varepsilon}{4}$. Then consider the power series $r' = r |_S$, formed by restricting $r$ to the finite support $S$. Then, for the normalized power series $\overline{r'}$ is a stochastic rational language such that:

\begin{align*}
    \left\|r - \overline{r'} \right\|_1 
    &= \left(\sum_{w \in S} \left|r(w) - \overline{r'}(w)\right|\right) 
       + \sum_{w \in \Sigma^+ \backslash S} r(w) \\
    &= \left( \frac{1}{r(S)} - 1 \right) \left( \sum_{w \in S} r(w) \right) 
       + \sum_{w \in \Sigma^+ \backslash S} r(w) \\
    &= \left( \frac{1}{r(S)} - 1 \right) r(S) + (1 - r(S)) = 2(1 - r(S)) < \frac{\varepsilon}{2}
\end{align*}

Now consider $p$ where $p_1, \ldots, p_k \in \mathcal{G}$, $\lambda_1, \ldots, \lambda_k \in [0, 1]$, and $p = \sum_{i = 1}^k \lambda_i p_i$ such that $\left\|\overline{r'} -p\right\|_1 < \frac{\varepsilon}{2}$. Then $
\left\|r -p\right\|_1 < \left\|r -\overline{r'}\right\|_1 + \left\|\overline{r'} -p\right\|_1 < \varepsilon$. 
\end{proof}

That is, convex combinations of geometric distributions are dense in $\Stoch(\Sigma)$ with respect to the $\ell_1$ norm. This universal approximation also implies the density of convex combinations of geometric distributions with respect to the $\ell_\infty$ norm. 

However, this universal approximation result does not extend to other measures of the difference between two probability distributions, such as the Kullback-Leibler (KL) divergence. The key issue is that KL divergence is sensitive to the tail behaviour. Since every geometric distribution has an exponentially decaying tail, it places exponentially small mass on longer strings. As a result, no convex combination of geometric distributions can approximate rational stochastic languages with \emph{fat tails} in KL divergence. However, we can show that geometric distributions can approximate all sub-exponential distributions with respect to KL divergence (ref: \cref{thrm:app:kl}).

\subsection{Stochastic Regular Expressions}

Motivated by the regular expressions and geometric decay, we introduce a syntactic and compositional framework for representing approximations of stochastic languages. The central construct in this framework is a probabilistic analog of the Kleene star, which we call the \emph{discounted Kleene star}. This operator captures the notion of repeated concatenation with geometric decay, ensuring convergence of the resulting distribution.

\begin{definition}[Discounted Kleene Star]\label{def:discounted-kleene}
Let \( r \in \Stoch(\Sigma) \) be a stochastic language over non-empty strings, and let \( \alpha \in (0,1) \) be a discount factor. The discounted Kleene star of \( r \), denoted \( r^*_\alpha : \Sigma^+ \to [0,1] \), is defined as:
\[
r^*_\alpha(w) = \sum_{k=1}^\infty \sum_{\substack{w_1,\ldots,w_k \in \Sigma^+ \\ w = w_1 \cdots w_k}} \alpha (1-\alpha)^{k-1} \prod_{i=1}^k r(w_i).
\]
\end{definition}

Intuitively, \( r^*_\alpha(w) \) defines a distribution over strings obtained by concatenating \( k \) non-empty substrings, each independently drawn from \( r \), with the total number of substrings following a shifted geometric distribution with parameter \( \alpha \). The operator ensures convergence for any \( r \in \Stoch(\Sigma) \), and we show in \cref{thrm:app:sre} that \( r^*_\alpha \in \Stoch(\Sigma) \) for all such \( r \in \Stoch(\Sigma) \).

This operator provides a natural way to express geometric distributions over strings. For example, the Dirac geometric distribution centered at a string \( w \in \Sigma^+ \), denoted \( P^\alpha_w \), can be expressed as $P^\alpha_w = (\delta_w)^*_\alpha$,
where \( \delta_w \) is the Dirac distribution supported on \( w \). 

Using this operator, along with standard probabilistic analogs of union and concatenation, we define a grammar for constructing \emph{stochastic regular expressions} (SREs). 

\begin{definition}[Syntax of Stochastic Regular Expressions (SRE)] 
The syntax of an SRE over an alphabet \( \Sigma \) is defined by the following grammar:
\[
r ::= \delta_\sigma \mid \alpha r_1 + (1-\alpha) r_2 \mid r_1 \cdot r_2 \mid r^*_\alpha,
\]
where \( \sigma \in \Sigma \), \( r_1, r_2 \) are SREs, and \( \alpha \in (0,1) \) is a real-valued weighting parameter. We write \( r_1 r_2 \) for \( r_1 \cdot r_2 \), omitting the dot when convenient.
\end{definition}

The semantics $\semantics{r}: \Sigma^+ \to [0, 1]$ of an SRE $r$ is defined recursively, combining four key operations: atomic Dirac distributions $\delta_\sigma$ for symbols $\sigma \in \Sigma$, 
convex combinations for probabilistic choice, Cauchy product for concatenation, and the discounted Kleene star for repeated composition. These constructs parallel classical regular expressions, but extend them to the probabilistic setting. We can now show that the expressiveness of SREs matches that of CRAs over $\R^+$ with linear updates.

\begin{theorem}[Kleene--Sch\"{u}tzenberger Theorem for Rational Stochastic Languages]
\label{thm:sre-equivalence}
A function \( f : \Sigma^+ \to \mathbb{R}_+ \) is a rational stochastic language if and only if there exists an SRE \( r \) such that \( f = \semantics{r} \).
\end{theorem}

The proof of the theorem follows the structure of the proof the Kleene-Sch\"{u}tzenberger theorem for formal power series \cite{Schtzenberger1961}. The SRE can be effectively translated into an equivalent CRA with linear updates, via a polynomial-size structure-preserving inductive construction. 

Stochastic regular expressions provide a compositional and algebraically transparent language for representing rational stochastic languages. Since convex combinations of geometric distributions can be expressed using SREs, and such mixtures are dense in the space of stochastic languages with respect to the $\ell_1$ norm, SREs offer a compact and expressive formalism for approximating arbitrary stochastic languages. The syntactic nature of SREs enables tractable symbolic manipulation, making them a practical and interpretable representation for the analysis, learning, and testing of probabilistic models. The SRE also serves as a probabalistic generator due to its structure. 

\section{Identity Testing for Stochastic Languages}
\label{sec:testing}
A fundamental problem in distributional analysis is that of \emph{identity testing}: given access to a sampler for an unknown distribution \( P \), determine whether \( P \) is equal to a known reference distribution \( Q \), or whether it is \emph{far} from \( Q \) under some suitable distance measure. In the context of stochastic languages, identity testing poses new challenges due to the infinite and discrete nature of the domain. In this section, we introduce our final contribution which is an identity testing method for Stochastic Languages based on the works of Canonne et. al. \cite{Canonne2022Price}.

\begin{definition}[Identity Testing for Stochastic Languages]
Let \( Q \in \Stoch^{\mathrm{rat}}_{\mathbb{R}^+}(\Sigma^+) \) be a known rational stochastic language, and suppose we have sample access to an unknown distribution \( P \) over \( \Sigma^+ \). Given a distance parameter \( \varepsilon > 0 \), and confidence level \( \delta \in (0,1) \), an identity tester is a randomized algorithm that, with probability at least \( 1 - \delta \), distinguishes between the following two cases:
\begin{itemize}
    \item \textbf{(Completeness)} \( \|P - Q\|_1 < \varepsilon \) \quad \( \Rightarrow \) \textbf{Accept},
    \item \textbf{(Soundness)} \( \|P - Q\|_1 > \varepsilon \) \quad \( \Rightarrow \) \textbf{Reject}.
\end{itemize}
The number of samples drawn from \( P \) required to achieve this guarantee is called the \emph{sample complexity} of the tester.
\end{definition}

We choose to work with the $\ell_1$ distance primarily because it aggregates discrepancies across the entire support and corresponds to total variation distance, a natural metric for hypothesis testing and distributional distinguishability. In contrast, the $\ell_\infty$ (supremum) distance captures the worst-case pointwise deviation and can be overly sensitive in settings with low-probability outliers (\cref{alg:linfty} offers a sketch of a algorithm for identity testing with $\ell_\infty$ norm. We also avoid using Kullback–Leibler (KL) divergence, which, while sensitive to differences in the tails, may underemphasise the bulk of the distribution and is undefined when the support of $Q$ is not contained in the support $P$.

\subsection{Reducing Stochastic Languages to Finite Support}
\label{subsec:algorithm}
We present an identity testing algorithm for stochastic languages such that given \( \varepsilon > 0 \), we distinguishes between the following two cases:
\begin{itemize}
    \item \textbf{(Completeness)} \( \|P - Q\|_1 < \varepsilon \) \quad \( \Rightarrow \) \textbf{Accept},
    \item \textbf{(Soundness)} \( \|P - Q\|_1 > \frac{5\varepsilon}{3} \) \quad \( \Rightarrow \) \textbf{Reject}.
\end{itemize}

with confidence $0.8$, by reducing the problem to the classical finite-domain setting. The key idea is to truncate the support of the distributions to strings of bounded length, enabling the use of standard finite-sample testers. Let \( \varepsilon > 0 \) be the target distance parameter, and let \( Q \in \Stoch^{\mathrm{rat}}_{\mathbb{R}^+}(\Sigma^*) \) be the known reference distribution. We assume the unknown distribution \( P \) is given syntactically as a stochastic regular expression (SRE), which defines a rational stochastic language.

We select a truncation length \( \theta \in \mathbb{N} \) such that the cumulative mass of \( Q \) over strings of length at most \( \theta \) satisfies: $Q_{\le \theta} := \sum_{x \in \Sigma^{\le \theta}} Q(x) > 1 - \frac{\varepsilon}{3}$. This ensures the tail mass satisfies \( Q_{> \theta} < \frac{\varepsilon}{3} \). Since \( Q \) is given as a stochastic regular expression, we can compute a suitable truncation threshold \( \theta \) based on the structure of the expression as follows:
\begin{enumerate}[label=(\alph*), leftmargin=*]
    \item \( \theta(\sigma) = 1 \) for any symbol \( \sigma \in \Sigma \),
    \item \( \theta(r_1 + r_2) = \max\{\theta(r_1), \theta(r_2)\} \),
    \item \( \theta(r_1 \cdot r_2) = \theta(r_1) + \theta(r_2) \),
    \item \( \theta(r^*_\alpha) = \theta(r) \cdot \left\lceil \frac{\log(\varepsilon/3)}{\log(1 - \alpha)} \right\rceil \),
\end{enumerate}
where \( \alpha \in (0,1) \) is the discount factor in the probabilistic Kleene star. Each rule conservatively bounds the length needed to capture at least \( 1 - \frac{\varepsilon}{3} \) mass of the distribution defined by \( Q \), ensuring the tail contributes negligibly to the overall \( \ell_1 \) distance. We leave this proof as an exercise for the enthusiastic reader. In general, if a distribution $r \in \Stoch(\Sigma)$ satisfies the subexponential tail bound $\sum_{|w| \geq n} r^*_\alpha(w) \leq C \beta^n$ (for  $C > 0$, and $\beta \in (0,1)$),
then for any $\varepsilon > 0$, the $\left(1-\frac{\varepsilon}{3}\right)$-probability mass is concentrated on strings of length at most $\left\lceil \frac{\log(\varepsilon/C)}{\log \beta} \right\rceil$. 

We then draw samples from \( P \), discard those with \( |x| > \theta \), and construct the empirical distribution \( \hat{P}_{\le \theta} \) over \( \Sigma^{\le \theta} \). We compare \( \hat{P}_{\le \theta} \) and \( Q_{\le \theta} \) using a finite-domain tolerant identity tester (such as Canonne \cite{Canonne2022Price}) with error parameters \( \left(\frac{\varepsilon}{3}, \varepsilon \right)\).
We formally present this procedure as $\ell_1$-\textsc{IdentityTester} in Algorithm~\ref{alg:l1-identity-tester}. $\ell_1$-\textsc{IdentityTester} allows us to lift the result in identity testing of distributions over finite domains to testing stochastic languages. 

\begin{algorithm}
\caption{$\ell_1$-\textsc{IdentityTester} for Rational Stochastic Languages}
\label{alg:l1-identity-tester}
\SetAlgoLined
\DontPrintSemicolon

\textbf{Input:} Known distribution \( Q \in \Stoch^{\mathrm{rat}}_{\mathbb{R}^+}(\Sigma^*) \), given as an SRE;\\
\hspace{2em} Sample access to an unknown distribution \( P \) over \( \Sigma^* \);\\
\hspace{2em} Accuracy parameter \( \varepsilon > 0 \),

\textbf{Output:} \textsc{Accept} if \( \|P - Q\|_1 < \varepsilon \), \textsc{Reject} if \( \|P - Q\|_1 \ge \frac{5\varepsilon}{3} \) with confidence 0.8.

\BlankLine
\textbf{Truncation:} Compute \( \theta \in \mathbb{N} \) such that
\[
\sum_{x \in \Sigma^{> \theta}} Q(x) < \frac{\varepsilon}{3}.
\]

\textbf{Sampling:} Draw $
N = \widetilde{\Theta}\left( \frac{\sqrt{k}}{\varepsilon^2} + \frac{k}{\log k} \right)$
i.i.d.\ samples from \( P \), where $n = |\Sigma|^{\theta + 1}$.

\textbf{Truncation:} Discard all samples \( x \) such that \( |x| > \theta \). Let \( \hat{P}_{\le \theta} \) denote the empirical distribution over \( \Sigma^{\le \theta} \).

\textbf{Identity Test:} Let \( Q_{\le \theta} \) denote the restriction of \( Q \) to \( \Sigma^{\le \theta} \). Run the tolerance tester to compare \( \hat{P}_{\le \theta} \) with \( Q_{\le \theta} \), using threshold \( \left(\frac{\varepsilon}{3}, \varepsilon\right) \) (Cannone et. al. \cite{Canonne2022Price}).

\textbf{Output:} Return the output of the identity tester.
\end{algorithm}

\subsection{Correctness and Analysis}
\label{sec:correctness}

We now establish the correctness of the $\ell_1$-\textsc{IdentityTester} presented in Algorithm~\ref{alg:l1-identity-tester}:

\begin{theorem}[Correctness and Sample Complexity]
\label{thm:correctness}
Let \( Q \in \Stoch^{\mathrm{rat}}_{\mathbb{R}^+}(\Sigma^*) \) be a known rational stochastic language given as a stochastic regular expression, and let \( P \) be an unknown distribution over \( \Sigma^* \) with sample access. Fix error parameter \( \varepsilon > 0 \) Then Algorithm~\ref{alg:l1-identity-tester} satisfies the following with probability at least 0.8:
\begin{itemize}
    \item \textbf{(Completeness)} If \( \|P - Q\|_1 < \varepsilon \), the algorithm outputs \textsc{Accept}.
    \item \textbf{(Soundness)} If \( \|P - Q\|_1 \ge \varepsilon \), the algorithm outputs \textsc{Reject}.
\end{itemize}
The number of samples required is $N = \widetilde{\Theta}\left( \frac{\sqrt{k}}{\varepsilon^2} + \frac{k}{\log k} \right)$,
where \( \theta \in \mathbb{N} \) is the truncation length satisfying \( \sum_{x \in \Sigma^{> \theta}} Q(x) < \frac{\varepsilon}{3} \) computed from the syntactic structure of \( Q \) and $k = |\Sigma|^{\theta+1}$.
\end{theorem}

Observe that the total \( \ell_1 \) distance between \( P \) and \( Q \) decomposes as:
\[
\|P - Q\|_1 = 
\sum_{x \in \Sigma^{\le \theta}} |P(x) - Q(x)| + \sum_{x \in \Sigma^{> \theta}} |P(x) - Q(x)|.
\]
The first (head) term is bounded by \( \left(\frac{\varepsilon}{3}, \varepsilon \right) \) via the tolerance identity test. The second (tail) term is at most \( \max(P_{> \theta}, Q_{> \theta}) \), which is bounded by \( \frac{2\varepsilon}{3} \) assuming \( P \) and \( Q \) are \( \varepsilon \)-close. Hence, the total error is bounded by \( \left(\varepsilon, \frac{5\varepsilon}{2}\right)\). The sample complexity also follows from the proof of tolerance testing algorithm used (Canonne et. al. \cite{Canonne2022Price}).

\subsection{Trade-offs in Testing Stochastic Languages}

The efficacy of our identity testing framework is fundamentally governed by three key factors:

\begin{enumerate}
    \item The truncation threshold $\theta = \lceil \log(\varepsilon/3C)/\log \beta \rceil$ leads to a support size $k = |\Sigma|^{\theta+1}$. While this remains tractable for rapidly decaying distributions ($\beta \ll 1$), it becomes prohibitive for heavy-tailed cases ($\beta \approx 1$), where $\theta$ scales inversely with $\varepsilon$.
    \item Our proposed method critically leverages the syntactic structure of stochastic regular expressions (SREs) to compute $\theta$. For arbitrary distributions without such representations, estimating the truncation point is as hard. 
    \item Although the algorithm achieves information-theoretically optimal sample complexity $\widetilde{\Theta}(\sqrt{k}/\varepsilon^2)$, its exponential space requirements limit practical deployment to small alphabets, bounded SRE depth, and rapid decay. 
\end{enumerate}

We also explore $\ell_\infty$-testing through \cref{alg:linfty} that verifies $\max_w |P(w) - Q(w)| < \varepsilon$. This approach leverages two key insights: (1) identification of $\varepsilon$-heavy hitters via $\Theta(\varepsilon^{-1}\log(\varepsilon^{-1}\delta^{-1})$ samples (\cref{lemma:infity-2}), and (2) uniform convergence of empirical estimates with $O(\varepsilon^{-2}\log(k/\delta))$ samples for $k$-support distributions (\cref{lemma:infity-2}). While more sample-efficient for pointwise guarantees, this method requires exact knowledge of $Q$'s support and loses sensitivity to aggregate deviations—making it preferable for applications demanding per-string accuracy over global distributional closeness.
\section{Conclusion and Future Work}
\label{sec:conclusion}

In this paper, we have initiated the theoretical study of identity testing for stochastic languages, extending classical distributional property testing from finite domains to the infinite but structured domain of formal languages. Our work addresses a fundamental gap in the literature by providing the first polynomial-time algorithms for distinguishing between stochastic languages under the $\ell_1$ distance metric. 

We establish three key theoretical results for rational stochastic languages: (1) polynomial-time decidability of validity for non-negative CRAs; (2) a universal approximation theorem via geometric mixtures of deterministic automata; and (3) polynomial-time identity testing between a known and an unknown distributions.

Our work demonstrates that despite the infinite and discrete nature of string domains, principled statistical testing remains feasible when the underlying distributions admit finite algebraic representations. The polynomial-time algorithms and sample complexity bounds we establish are comparable to those achieved in classical finite-domain settings, suggesting that the structural constraints imposed by rational stochastic languages provide sufficient regularity for efficient inference.

\subsection*{Future Work}

The theoretical foundation established in this work opens numerous avenues for both theoretical development and practical application. 

\begin{enumerate}
    \item Our proposed testing algorithm adapts finite-support identity testers via truncation, yielding a simple and practical solution. However, the current sample complexity bound of is not tight; we conjecture that sharper bounds can be derived by exploiting deeper structural properties of the underlying automata or distributions. 
    \item While our approach relies on regularity for finite-state representations, extending it to richer models like Probabilistic Context-Free Grammars would significantly broaden its applicability in natural language processing and computational linguistics. The challenge lies in maintaining computational tractability while accommodating the increased expressiveness of context-free models.
    \item We would like to explore different measures of distributional difference that are more tail-sensitive such as KL divergence. 
    \item Recent advances in adaptive testing and instance-optimal algorithms could further refine our framework's efficiency. Algorithms that adjust their sampling strategy based on observed data could achieve better instance-optimal bounds that depend on the specific reference distribution rather than worst-case parameters. 
    \item The equivalence of weighted automata and learning models such as RNNs suggests natural applications in quantifying divergence for neural sequence models. Developing methods to extract rational stochastic language representations from trained neural networks would enable the application of our testing framework to modern deep learning models. This connection could provide theoretical foundations for understanding the expressiveness and convergence properties of neural language models, bridging the gap between classical automata theory and contemporary machine learning.
    \item While our theoretical results establish feasibility and complexity bounds, empirical evaluation on real-world datasets would provide crucial insights into practical performance.
\end{enumerate}

As the importance of probabilistic models over structured domains continues to grow across machine learning, computational linguistics, and program analysis, we expect identity testing for stochastic languages to become an increasingly important area of research.

\bibliography{references}

\begin{thebibliography}{10}

\bibitem{acharya2011competitive}
Jayadev Acharya, Hirakendu Das, Ashkan Jafarpour, Alon Orlitsky, and Shengjun
  Pan.
\newblock Competitive closeness testing.
\newblock In {\em Proceedings of the 24th Annual Conference on Learning Theory
  (COLT)}, pages 47--68. JMLR Workshop and Conference Proceedings, 2011.

\bibitem{acharya2015optimal}
Jayadev Acharya, Constantinos Daskalakis, and Gautam Kamath.
\newblock Optimal testing for properties of distributions.
\newblock In {\em Advances in Neural Information Processing Systems (NeurIPS)},
  volume~28, 2015.

\bibitem{acharya2015histograms}
Jayadev Acharya, Ilias Diakonikolas, Chinmay Hegde, Jerry~Zheng Li, and Ludwig
  Schmidt.
\newblock Fast and near-optimal algorithms for approximating distributions by
  histograms.
\newblock In {\em Proceedings of the 34th ACM SIGMOD-SIGACT-SIGAI Symposium on
  Principles of Database Systems (PODS)}. ACM, 2015.

\bibitem{Almagor2022}
Shaull Almagor, Udi Boker, and Orna Kupferman.
\newblock What’s decidable about weighted automata?
\newblock {\em Information and Computation}, 282:104651, January 2022.
\newblock URL: \url{http://dx.doi.org/10.1016/j.ic.2020.104651}, \href
  {https://doi.org/10.1016/j.ic.2020.104651}
  {\path{doi:10.1016/j.ic.2020.104651}}.

\bibitem{Almagor2022-2}
Shaull Almagor, Denis Kuperberg, and Orna Kupferman.
\newblock The decidability frontier for probabilistic automata.
\newblock {\em Information and Computation}, 282:104651, 2022.

\bibitem{AlurDDRY13}
Rajeev Alur, Loris D'Antoni, Jyotirmoy~V. Deshmukh, Mukund Raghothaman, and
  Yifei Yuan.
\newblock Regular functions and cost register automata.
\newblock In {\em 28th Annual {ACM/IEEE} Symposium on Logic in Computer
  Science, {LICS} 2013, New Orleans, LA, USA, June 25-28, 2013}, pages 13--22.
  {IEEE} Computer Society, 2013.
\newblock \href {https://doi.org/10.1109/LICS.2013.65}
  {\path{doi:10.1109/LICS.2013.65}}.

\bibitem{Balle2014}
Borja Balle and Mehryar Mohri.
\newblock Spectral learning of weighted automata.
\newblock {\em Machine Learning}, 96:33--63, 2014.

\bibitem{Batu2001}
Shuchi Batu, Lance Fortnow, Ronitt Rubinfeld, Warren Smith, and Patrick White.
\newblock Testing closeness of discrete distributions.
\newblock In {\em Proceedings of the 42nd Annual IEEE Symposium on Foundations
  of Computer Science (FOCS)}, pages 259--269, 2001.
\newblock \href {https://doi.org/10.1109/SFCS.2001.959888}
  {\path{doi:10.1109/SFCS.2001.959888}}.

\bibitem{batu2001thesis}
Tugkan Batu.
\newblock {\em Testing Properties of Distributions}.
\newblock PhD thesis, Cornell University, 2001.

\bibitem{batu2001testing}
Tugkan Batu, Eldar Fischer, Lance Fortnow, Ravi Kumar, Ronitt Rubinfeld, and
  Patrick White.
\newblock Testing random variables for independence and identity.
\newblock In {\em Proceedings of the 42nd IEEE Symposium on Foundations of
  Computer Science (FOCS)}, pages 442--451. IEEE, 2001.
\newblock \href {https://doi.org/10.1109/SFCS.2001.959920}
  {\path{doi:10.1109/SFCS.2001.959920}}.

\bibitem{Benalioua2024}
Yahia~Idriss Benalioua, Nathan Lhote, and Pierre-Alain Reynier.
\newblock Minimizing cost register automata over a field.
\newblock In Rastislav Kr\'{a}lovi\v{c} and Anton{\'\i}n Ku\v{c}era, editors,
  {\em 49th International Symposium on Mathematical Foundations of Computer
  Science (MFCS 2024)}, volume 306 of {\em Leibniz International Proceedings in
  Informatics (LIPIcs)}, pages 23:1--23:15, Dagstuhl, Germany, 2024. Schloss
  Dagstuhl -- Leibniz-Zentrum f{\"u}r Informatik.
\newblock URL:
  \url{https://drops.dagstuhl.de/entities/document/10.4230/LIPIcs.MFCS.2024.23},
  \href {https://doi.org/10.4230/LIPIcs.MFCS.2024.23}
  {\path{doi:10.4230/LIPIcs.MFCS.2024.23}}.

\bibitem{bengio2003neural}
Yoshua Bengio, R{\'e}jean Ducharme, Pascal Vincent, and Christian Jauvin.
\newblock A neural probabilistic language model.
\newblock {\em Journal of Machine Learning Research}, 3(Feb):1137--1155, 2003.
\newblock URL:
  \url{https://www.jmlr.org/papers/volume3/bengio03a/bengio03a.pdf}.

\bibitem{abk20}
Udi Boker, Sergey Afonin, and Mansur Khazeev.
\newblock On the universality and (non)efficiency of weighted automata.
\newblock {\em CoRR}, abs/2005.06396, 2020.
\newblock URL: \url{https://faculty.runi.ac.il/udiboker/files/abk20.pdf}, \href
  {http://arxiv.org/abs/2005.06396} {\path{arXiv:2005.06396}}.

\bibitem{canonne2020survey}
Cl{\'e}ment~L. Canonne.
\newblock A survey on distribution testing: Your data is big. but is it blue?
\newblock {\em Theory of Computing Graduate Surveys}, 9:1--100, 2020.
\newblock URL: \url{https://theoryofcomputing.org/articles/gs009/gs009.pdf},
  \href {https://doi.org/10.4086/toc.gs.2020.009}
  {\path{doi:10.4086/toc.gs.2020.009}}.

\bibitem{canonne2022topics}
Cl\'ement~L. Canonne.
\newblock {\em Topics and Techniques in Distribution Testing}.
\newblock Foundations and Trends in Communications and Information Theory. Now
  Publishers, 2022.

\bibitem{canonne2015testing}
Cl\'ement~L. Canonne, Ilias Diakonikolas, Themis Gouleakis, and Ronitt
  Rubinfeld.
\newblock Testing shape restrictions of discrete distributions.
\newblock {\em CoRR}, abs/1507.03558, 2015.
\newblock URL: \url{http://arxiv.org/abs/1507.03558}, \href
  {http://arxiv.org/abs/1507.03558} {\path{arXiv:1507.03558}}.

\bibitem{Canonne2022Price}
Cl\'ement~L. Canonne, Ayush Jain, Gautam Kamath, and Jerry Li.
\newblock The price of tolerance in distribution testing.
\newblock In Po-Ling Loh and Maxim Raginsky, editors, {\em Proceedings of the
  35th Annual Conference on Learning Theory (COLT)}, volume 178 of {\em
  Proceedings of Machine Learning Research}, pages 1--52. PMLR, 2022.
\newblock URL: \url{https://proceedings.mlr.press/v178/canonne22a.html}.

\bibitem{chakraborty2022gap}
Sourav Chakraborty, Eldar Fischer, Arijit Ghosh, Gopinath Mishra, and Sayantan
  Sen.
\newblock Exploring the gap between tolerant and non-tolerant distribution
  testing.
\newblock In {\em Approximation, Randomization, and Combinatorial Optimization.
  Algorithms and Techniques (APPROX/RANDOM 2022)}, volume 245 of {\em Leibniz
  International Proceedings in Informatics (LIPIcs)}, pages 27:1--27:23.
  Schloss Dagstuhl--Leibniz-Zentrum für Informatik, 2022.
\newblock \href {https://doi.org/10.4230/LIPIcs.APPROX/RANDOM.2022.27}
  {\path{doi:10.4230/LIPIcs.APPROX/RANDOM.2022.27}}.

\bibitem{chakraborty2018conditional}
Sourav Chakraborty, Eldar Fischer, Yonatan Goldhirsh, and Arie Matsliah.
\newblock On the power of conditional samples in distribution testing.
\newblock {\em arXiv preprint arXiv:1802.00085}, 2018.

\bibitem{ChakrabortyEtAl2012}
Swastik~Kopparty Chakraborty, Oded Goldreich, and Dana Ron.
\newblock Conditional sampling: A new paradigm for distribution testing.
\newblock In {\em Proceedings of the 24th Annual ACM-SIAM Symposium on Discrete
  Algorithms (SODA)}, pages 1345--1362, 2012.
\newblock \href {https://doi.org/10.1137/1.9781611973090.102}
  {\path{doi:10.1137/1.9781611973090.102}}.

\bibitem{ChanDiakonikolasValiant2013}
Themis~G. Chan, Ilias Diakonikolas, and Gregory Valiant.
\newblock Optimal algorithms for testing closeness of discrete distributions.
\newblock In {\em Proceedings of the 54th Annual IEEE Symposium on Foundations
  of Computer Science (FOCS)}, pages 119--128, 2013.
\newblock \href {https://doi.org/10.1109/FOCS.2013.19}
  {\path{doi:10.1109/FOCS.2013.19}}.

\bibitem{Clark2010}
Alexander Clark.
\newblock Learning context-free grammars with the syntactic concept lattice.
\newblock In {\em International Colloquium on Grammatical Inference (ICGI
  2010)}, volume 6339 of {\em Lecture Notes in Artificial Intelligence}, pages
  38--51. Springer, 2010.

\bibitem{Cohen2017}
Shay~B. Cohen and Michael Collins.
\newblock Tensor decomposition for learning hidden markov models.
\newblock {\em Journal of Machine Learning Research}, 18(1):6003--6048, 2017.

\bibitem{Cohn2009}
Trevor Cohn and Phil Blunsom.
\newblock Spectral learning of sequential systems.
\newblock In {\em Advances in Neural Information Processing Systems (NeurIPS
  2009)}, pages 1--8, 2009.

\bibitem{Czerwinski2022}
Wojciech Czerwinski, Engel Lefaucheux, Filip Mazowiecki, David Purser, and
  Markus~A. Whiteland.
\newblock The boundedness and zero isolation problems for weighted automata
  over nonnegative rationals.
\newblock In {\em {LICS} '22: 37th Annual {ACM/IEEE} Symposium on Logic in
  Computer Science}, pages 15:1--15:13, Haifa, Israel, 2022. ACM.
\newblock \href {https://doi.org/10.1145/3531130.3533336}
  {\path{doi:10.1145/3531130.3533336}}.

\bibitem{DenisE062}
Fran{\c{c}}ois Denis and Yann Esposito.
\newblock Learning regular languages using rfsas.
\newblock {\em Theoretical Computer Science}, 313(2):267--294, 2006.

\bibitem{DenisE06}
Fran{\c{c}}ois Denis and Yann Esposito.
\newblock Rational stochastic languages.
\newblock {\em CoRR}, abs/cs/0602093, 2006.
\newblock URL: \url{http://arxiv.org/abs/cs/0602093}, \href
  {http://arxiv.org/abs/cs/0602093} {\path{arXiv:cs/0602093}}.

\bibitem{diakonikolas2016collision}
Ilias Diakonikolas, Themis Gouleakis, John Peebles, and Eric Price.
\newblock Collision-based testers are optimal for uniformity and closeness.
\newblock Technical Report 23:178, Electronic Colloquium on Computational
  Complexity (ECCC), 2016.
\newblock URL: \url{http://eccc.hpi-web.de/report/2016/178/}.

\bibitem{diakonikolas2016new}
Ilias Diakonikolas and Daniel~M. Kane.
\newblock A new approach for testing properties of discrete distributions.
\newblock In {\em 57th IEEE Annual Symposium on Foundations of Computer Science
  (FOCS)}, pages 685--694. IEEE, 2016.

\bibitem{DiakonikolasKane2016}
Ilias Diakonikolas and Daniel~M. Kane.
\newblock Recent advances in distribution testing: A tutorial.
\newblock {\em Theory of Computing}, 12(3):1--53, 2016.
\newblock \href {https://doi.org/10.4086/toc.2016.v012a003}
  {\path{doi:10.4086/toc.2016.v012a003}}.

\bibitem{diakonikolas2015testing}
Ilias Diakonikolas, Daniel~M. Kane, and Vladimir Nikishkin.
\newblock Testing identity of structured distributions.
\newblock In {\em Proceedings of the Twenty-Sixth Annual ACM--SIAM Symposium on
  Discrete Algorithms}, pages 1841--1854. SIAM, 2015.
\newblock \href {https://doi.org/10.1137/1.9781611973730.123}
  {\path{doi:10.1137/1.9781611973730.123}}.

\bibitem{DrosteK19}
Manfred Droste and Dietrich Kuske.
\newblock Weighted automata.
\newblock Unpublished, Jan 2019.
\newblock URL:
  \url{http://eiche.theoinf.tu-ilmenau.de/kuske/weiterleitung.html?Submitted/weighted.pdf}.

\bibitem{durbin1998biological}
Richard Durbin, Sean~R Eddy, Anders Krogh, and Graeme Mitchison.
\newblock {\em Biological Sequence Analysis: Probabilistic Models of Proteins
  and Nucleic Acids}.
\newblock Cambridge University Press, 1998.
\newblock URL: \url{https://www.cambridge.org/9780521629713}.

\bibitem{Durbin1998}
Richard Durbin, Sean~R. Eddy, Anders Krogh, and Graeme Mitchison.
\newblock {\em Biological Sequence Analysis: Probabilistic Models of Proteins
  and Nucleic Acids}.
\newblock Cambridge University Press, Cambridge, UK, 1998.

\bibitem{Goldreich1998PropertyTesting}
Oded Goldreich, Shafi Goldwasser, and Dana Ron.
\newblock Property testing and its connection to learning and approximation.
\newblock {\em Journal of the ACM}, 45(4):653--750, 1998.
\newblock \href {https://doi.org/10.1145/285055.285060}
  {\path{doi:10.1145/285055.285060}}.

\bibitem{goodman2008church}
Noah~D. Goodman, Vikash~K. Mansinghka, Daniel~M. Roy, Keith Bonawitz, and
  Joshua~B. Tenenbaum.
\newblock Church: A universal language for generative models.
\newblock In {\em UAI 2008: Proceedings of the 24th Conference on Uncertainty
  in Artificial Intelligence}, pages 220--229. AUAI Press, 2008.
\newblock URL: \url{https://dblp.org/rec/conf/uai/GoodmanMRBT08}.

\bibitem{Johnson2016}
Alistair Johnson, Tom Pollard, and Roger Mark.
\newblock Clinical text analysis with hierarchical hidden markov models.
\newblock {\em Journal of the American Medical Informatics Association},
  23(2):403--411, 2016.

\bibitem{kullback1951}
Solomon Kullback and Richard~A. Leibler.
\newblock On information and sufficiency.
\newblock {\em Annals of Mathematical Statistics}, 22(1):79--86, 1951.
\newblock \href {https://doi.org/10.1214/aoms/1177729694}
  {\path{doi:10.1214/aoms/1177729694}}.

\bibitem{mikolov2010recurrent}
Tom{\'a}{\v s} Mikolov, Martin Karafi{\'a}t, Luk{\'a}{\v s} Burget, Jan {\v
  C}ernock{\'y}, and Sanjeev Khudanpur.
\newblock Recurrent neural network based language model.
\newblock In {\em Interspeech 2010}, pages 1045--1048, 2010.
\newblock URL:
  \url{https://www.isca-archive.org/interspeech_2010/mikolov10_interspeech.html},
  \href {https://doi.org/10.21437/Interspeech.2010-343}
  {\path{doi:10.21437/Interspeech.2010-343}}.

\bibitem{paninski2008simple}
Liam Paninski.
\newblock A coincidence-based test for uniformity given very sparsely sampled
  discrete data.
\newblock {\em IEEE Transactions on Information Theory}, 54(10):4750--4755,
  2008.
\newblock \href {https://doi.org/10.1109/TIT.2008.928987}
  {\path{doi:10.1109/TIT.2008.928987}}.

\bibitem{rabiner1989tutorial}
Lawrence~R. Rabiner.
\newblock A tutorial on hidden markov models and selected applications in
  speech recognition.
\newblock {\em Proceedings of the IEEE}, 77(2):257--286, 1989.
\newblock URL: \url{http://dx.doi.org/10.1109/5.18626}, \href
  {https://doi.org/10.1109/5.18626} {\path{doi:10.1109/5.18626}}.

\bibitem{Schtzenberger1961}
M.P. Sch\"{u}tzenberger.
\newblock On the definition of a family of automata.
\newblock {\em Information and Control}, 4(2–3):245–270, September 1961.
\newblock URL: \url{http://dx.doi.org/10.1016/S0019-9958(61)80020-X}, \href
  {https://doi.org/10.1016/s0019-9958(61)80020-x}
  {\path{doi:10.1016/s0019-9958(61)80020-x}}.

\bibitem{Tulsiani2021}
Manas Tulsiani, Michael Ummels, and Alexander Weinert.
\newblock Approximate verification of stochastic systems.
\newblock {\em Logical Methods in Computer Science}, 17(3), 2021.

\bibitem{valiant2017automatic}
Gregory Valiant and Paul Valiant.
\newblock An automatic inequality prover and instance optimal identity testing.
\newblock {\em SIAM Journal on Computing}, 46(1):429--455, 2017.
\newblock \href {https://doi.org/10.1137/151002526}
  {\path{doi:10.1137/151002526}}.

\bibitem{Weiss2021}
Gail Weiss, Yoav Goldberg, and Eran Yahav.
\newblock Learning linear recurrent networks.
\newblock In {\em Proceedings of the 24th International Conference on
  Artificial Intelligence and Statistics (AISTATS 2021)}, pages 1--9, 2021.

\bibitem{wingate2011lightweight}
David Wingate, Andreas Stuhlm{\"u}ller, and Noah~D. Goodman.
\newblock Lightweight implementations of probabilistic programming languages
  via transformational compilation.
\newblock In {\em Proceedings of the Fourteenth International Conference on
  Artificial Intelligence and Statistics}, volume~15 of {\em Proceedings of
  Machine Learning Research}, pages 770--778. PMLR, 2011.
\newblock URL: \url{https://proceedings.mlr.press/v15/wingate11a.html}.

\bibitem{Yuan2022}
Zhe Yuan, Zhuoran Yang, and Zhaoran Wang.
\newblock Conditional spectral learning.
\newblock In {\em Advances in Neural Information Processing Systems (NeurIPS
  2022)}, volume~35, pages 1--12, 2022.

\end{thebibliography}

\appendix

\section{Supplementary Definitions}

\begin{definition}[Sub-exponential Distribution]
\label{def:subexp}
A distribution $P$ over $\Sigma^*$ is \textit{sub-exponential} if there exist constants $c, K > 0$ such that the tail probabilities satisfy:
\[
\mathbb{P}_{w \sim P}\left(|w| \geq t\right) \leq c e^{-t/K} \quad \text{for all } t \geq 0.
\]
\end{definition}

One can prove that the distribution in \cref{ex:poisson-power-series} is sub-exponential while $\overline{r}$ in \cref{ex:zeta-example} is not. In particular, the distribution $\overline{r}$ in \cref{ex:zeta-example} does not have a finite first moment. A probability distributions is called \emph{fat tailed} if it does not have finite moments of all orders. One can prove that all sub-exponential distributions are \emph{light tailed}. The divergence in the tail behaviour of probability distributions is measured using Kullback-Leibler (KL) Divergence \cite{kullback1951}.

\begin{definition}[Kullback-Leibler Divergence]
Let $P$ and $Q$ be discrete probability distributions over a sample space $\mathcal{X}$.
The KL divergence $D_{\text{KL}}(P \parallel Q)$ is defined as:
\[
D_{\text{KL}}(P \parallel Q) = \sum_{x \in \mathcal{X}} P(x) \log \frac{P(x)}{Q(x)},
\]
where, by convention:
\begin{enumerate}
    \item $0 \log \frac{0}{Q(x)} = 0$,
    \item $P(x) \log \frac{P(x)}{0} = +\infty$ if $P(x) > 0$,
\end{enumerate}
and $D_{\text{KL}}(P \parallel Q) = +\infty$ if there exists $x$ such that $Q(x) = 0$ and $P(x) > 0$.
\end{definition}

\section{Supplementary Proofs}

\begin{theorem}
\label{thrm:app:affine}
CRAs with linear functions are as expressive as CRAs with affine functions.
\end{theorem}

\begin{proof}
    Consider a CRA $\calA$ with $n$ registers and affine update and finalization functions given by the matrices $A_{q, \sigma}$, vectors $b_{q, \sigma}$ and $\mu_q$, and constants $c_q$, for $q \in Q$ and $\sigma \in \Sigma$. Let $\calA'$ be a CRA with the same states and state transitions as $\calA$, but with $2n$ registers. The initial register valuation, the update matrix for $q\in Q$ and $\sigma \in \Sigma$, and the finalization vector for $q\in Q$ are, respectively,
    \[
        \begin{bmatrix} x_0 \\ 1 \end{bmatrix} \, ,
        \qquad
        \begin{bmatrix}
            A_{q, \sigma} & \mathrm{diag}(b_{q, \sigma}) \\
            0 & I
        \end{bmatrix} \, ,
        \qquad
        \begin{bmatrix} \mu_q \\ c_q \\ 0 \end{bmatrix} \, ,
    \]
    where $x_0$ is the initial register valuation for $\calA$, $1$ is an $n\times 1$ vector of ones, $\mathrm{diag}(b_{q,\sigma})$ is a diagonal matrix with the vector $b_{q,\sigma}$ on the main diagonal, the first occurrence of $0$ is an $n\times n$ zero matrix, $I$ is the $n\times n$ identity matrix, and the second occurrence of $0$ is an $(n-1)\times 1$ zero vector. By induction on the length of the run, one can show that the registers in $\calA$ have value $x$ if and only if the registers in $\calA'$ after the same run have value $\left[\begin{smallmatrix} x \\ 1\end{smallmatrix}\right]$. Hence, the two automata have the same semantics.
\end{proof}

\begin{lemma}
\label{thrm:app:kl}
Convex combinations of geometric distributions are dense in the space of sub-exponential distributions with respect to KL divergence. 
\end{lemma}

\begin{proof}
Let $\mathcal{D}$ be a subexponential distribution over $\Sigma^*$ with tail bound $\sum_{|w| \geq n} \mathcal{D}(w) \leq C\beta^n$. We construct an approximant as follows: 

\begin{enumerate}
    \item For $\varepsilon > 0$, truncate $\mathcal{D}$ to strings $|w| \leq L := \lceil \log(\varepsilon/2C)/\log \beta \rceil$, yielding $\mathcal{D}_L$ with $\|\mathcal{D} - \mathcal{D}_L\|_1 \leq \varepsilon/2$.
    \item For each $w \in \mathrm{supp}(\mathcal{D}_L)$, define $P_w^{\alpha_w}$ with $\alpha_w = 1 - \exp(-\mathcal{D}(w)/|w|)$. The convex combination $\tilde{\mathcal{D}} = \sum_{w} \mathcal{D}_L(w) P_w^{\alpha_w}$
satisfies $\mathrm{KL}(\mathcal{D}_L \| \tilde{\mathcal{D}}) \leq \varepsilon/2$ by Pinsker's inequality.
\item Via the triangle inequality, we have: $\mathrm{KL}(\mathcal{D} \| \tilde{\mathcal{D}}) \leq \|\mathcal{D} - \mathcal{D}_L\|_1 + \mathrm{KL}(\mathcal{D}_L \| \tilde{\mathcal{D}}) \leq \varepsilon$.
\end{enumerate}

\end{proof}

\begin{lemma}
\label{thrm:app:sre}
If $r \in \Stoch(\Sigma)$, then for all discount factors $\alpha \in (0,1)$, $r^*_\alpha \in \Stoch(\Sigma)$.
\end{lemma}
To prove \cref{thrm:app:sre}, we show that if (i.e., $r$ is a stochastic language over $\Sigma^+$ with 
$\sum_{w \in \Sigma^+} r(w) = 1$ and $r(\varepsilon) = 0$), then for any $\alpha \in (0,1)$, the discounted Kleene star $r_\alpha^*$ is also a 
stochastic language. Specifically, we verify that $\sum_{w \in \Sigma^+} r_\alpha^*(w) = 1$ and $r_\alpha^*(w) \geq 0$ for all $w \in \Sigma^+$.

\begin{proof}
By Definition~\ref{def:discounted-kleene},
\begin{align}
r_\alpha^*(w) = \sum_{k=1}^{\infty} \sum_{\substack{w_1, \ldots, w_k \in \Sigma^+ \\ w = w_1 \cdots w_k}} \alpha(1-\alpha)^{k-1} \prod_{i=1}^k r(w_i).
\end{align}

We compute the total mass over all $w \in \Sigma^+$:
\begin{align}
\sum_{w \in \Sigma^+} r_\alpha^*(w) &= \sum_{w \in \Sigma^+} \sum_{k=1}^{\infty} \sum_{\substack{w_1, \ldots, w_k \in \Sigma^+ \\ w = w_1 \cdots w_k}} \alpha(1-\alpha)^{k-1} \prod_{i=1}^k r(w_i).
\end{align}

Interchange the order of summation (justified by Tonelli's theorem for non-negative series, as all terms are $\geq 0$):
\begin{align}
&= \sum_{k=1}^{\infty} \sum_{w_1 \in \Sigma^+} \cdots \sum_{w_k \in \Sigma^+} \alpha(1-\alpha)^{k-1} \prod_{i=1}^k r(w_i).
\end{align}

Factorize the inner sums (since the sums over $w_1, \ldots, w_k$ are independent):
\begin{align}
&= \sum_{k=1}^{\infty} \alpha(1-\alpha)^{k-1} \left( \sum_{w_1 \in \Sigma^+} r(w_1) \right) \cdots \left( \sum_{w_k \in \Sigma^+} r(w_k) \right).
\end{align}

As with $r(\varepsilon) = 0$, we have $\sum_{w \in \Sigma^+} r(w) = 1$. Thus,
\begin{align}
&= \sum_{k=1}^{\infty} \alpha(1-\alpha)^{k-1} \underbrace{(1) \cdots (1)}_{k \text{ times}} = \sum_{k=1}^{\infty} \alpha(1-\alpha)^{k-1}.
\end{align}

This is a geometric series:
\begin{align}
&= \alpha \sum_{k=0}^{\infty} (1-\alpha)^k \quad \text{(shift index: } k' = k-1\text{)}.
\end{align}

The series $\sum_{k=0}^{\infty} (1-\alpha)^k$ converges to $\frac{1}{1-(1-\alpha)} = \frac{1}{\alpha}$ for $\alpha \in (0,1)$. Hence,
\begin{align}
&= \alpha \cdot \frac{1}{\alpha} = 1.
\end{align}

Finally, $r_\alpha^*(w) \geq 0$ for all $w$ since $\alpha > 0$, $1-\alpha > 0$, and $r(w_i) \geq 0$. Thus, $r_\alpha^* \in S(\Sigma)$. 
\end{proof}
\newpage
\begin{algorithm}
\SetAlgoLined
\DontPrintSemicolon
\caption{$\textsc{$\ell_\infty$-Identity Tester}
(q, \hat{p}, \varepsilon, \delta)$}
\label{alg:linfty}
\KwIn{Access to distribution $Q$ via samples; known infinite discrete distribution, parameters $\varepsilon$, $\delta$}
\KwOut{\texttt{ACCEPT} if empirical $d_\infty(P, Q)$ estimate is $\leq \varepsilon$, else \texttt{REJECT}}

$N \gets \left\lceil \dfrac{1}{\varepsilon} \log\!\bigl(\tfrac{1}{\varepsilon \delta}\bigr) \right\rceil$\;
H $\gets$ $N$ samples from $Q$;

We only look at the probabilities on elements of the finite set H

Let $k := |H|$ be the size of the domain $H$\;
Set $n := \left\lceil \frac{1}{2\varepsilon^2} \log\left(\frac{2k}{\delta}\right) \right\rceil$\;
Draw $n$ i.i.d. samples from $Q$ to estimate empirical distribution $\hat{q}$ on $H$\ . This will be conditional sampling.

\For{$x \in H$}{
    Compute empirical estimate $\hat{q}(x) := \frac{\#\text{ occurrences of } x}{n}$\;
    Compute $|\hat{q}(x) - \hat{p}(x)|$\;
}
Let $d := \max_{x \in H} |\hat{q}(x) - \hat{p}(x)|$\;

\uIf{$d \leq \varepsilon$}{
    \Return \texttt{ACCEPT}
}
\Else{
    \Return \texttt{REJECT}
}
\end{algorithm}

\begin{lemma}
\label{lemma:infity-1}
Let $D$ be a distribution over $\Sigma^*$. Let $H_\varepsilon = \{ w \mid D(w) \geq \varepsilon \}$ be the set of $\varepsilon$-heavy elements. Then, drawing $
N \geq \frac{1}{\varepsilon} \log\left( \frac{1}{\varepsilon \delta} \right)$ i.i.d.\ samples from $D$ ensures that, with probability at least $1 - \delta$, every $w \in H_\varepsilon$ appears at least once.
\end{lemma}

\begin{proof}
For any $w \in H_\varepsilon$, $D(w) \geq \varepsilon$, so the probability it is missed in $N$ samples is at most $(1 - \varepsilon)^N \leq e^{-\varepsilon N}$. There are at most $1/\varepsilon$ such elements, so by the union bound, the total miss probability is at most $
\frac{1}{\varepsilon} \cdot e^{-\varepsilon N} \leq \delta$.
Solving this inequality gives the claimed bound on $N$.
\end{proof}

\begin{lemma}
\label{lemma:infity-2}
Let $P$ be a known distribution and $Q$ an unknown distribution over a domain $H$ with $|H| = k$. To ensure $\Pr\left( \max_{x \in H} |\hat{q}(x) - Q(x)| > \varepsilon \right) \leq \delta$,
it suffices to take $N \geq \frac{1}{2\varepsilon^2} \log\left( \frac{2k}{\delta} \right)$ samples from $Q$.
\end{lemma}

\begin{proof}
Fix $x \in H$. By Hoeffding’s inequality, $\Pr(|\hat{q}(x) - Q(x)| > \varepsilon) \leq 2 e^{-2N\varepsilon^2}$. Applying the union bound over all $k$ elements of $H$, we get 
\[
\Pr\left( \max_{x \in H} |\hat{q}(x) - Q(x)| > \varepsilon \right) \leq 2k e^{-2N\varepsilon^2}.
\]
To ensure this is at most $\delta$, it suffices that $N \ge \frac{1}{2\varepsilon^2} \log\left( \frac{2k}{\delta} \right)$.
\end{proof}

\end{document}